\documentclass[11pt]{article}

\usepackage[letterpaper,margin=1in]{geometry}
\usepackage[T1]{fontenc}
\usepackage[utf8]{inputenc}
\usepackage{lmodern}
\usepackage{microtype}
\usepackage{amsmath,amssymb,amsthm,mathtools}
\usepackage{booktabs,array}
\usepackage{xcolor}
\usepackage{algorithm}
\usepackage{algpseudocode}
\usepackage{enumitem}
\usepackage{aliascnt}
\usepackage[numbers,sort&compress]{natbib}
\usepackage[colorlinks=true,allcolors=blue]{hyperref}
\usepackage[capitalize,nameinlink,noabbrev]{cleveref}
\usepackage{needspace}

\numberwithin{equation}{section}

\newtheorem{theorem}{Theorem}
\newaliascnt{lemma}{theorem}
\newtheorem{lemma}[lemma]{Lemma}
\aliascntresetthe{lemma}
\newaliascnt{proposition}{theorem}
\newtheorem{proposition}[proposition]{Proposition}
\aliascntresetthe{proposition}
\newaliascnt{corollary}{theorem}
\newtheorem{corollary}[corollary]{Corollary}
\aliascntresetthe{corollary}
\theoremstyle{definition}
\newaliascnt{definition}{theorem}
\newtheorem{definition}[definition]{Definition}
\aliascntresetthe{definition}
\theoremstyle{remark}
\newaliascnt{remark}{theorem}

\aliascntresetthe{remark}
\crefname{lemma}{Lemma}{Lemmas}
\Crefname{lemma}{Lemma}{Lemmas}
\crefname{proposition}{Proposition}{Propositions}
\Crefname{proposition}{Proposition}{Propositions}
\crefname{corollary}{Corollary}{Corollaries}
\Crefname{corollary}{Corollary}{Corollaries}
\crefname{definition}{Definition}{Definitions}
\Crefname{definition}{Definition}{Definitions}
\crefname{remark}{Remark}{Remarks}
\Crefname{remark}{Remark}{Remarks}
\usepackage{thm-restate}

\newcommand{\E}{\mathbb{E}}
\newcommand{\Prob}{\mathbb{P}}
\newcommand{\ind}{\mathbf{1}}
\newcommand{\TV}{d_{\mathrm{TV}}}
\newcommand{\dtv}{\TV}
\newcommand{\NACOND}{\mathsf{NACOND}}
\newcommand{\COND}{\mathsf{COND}}
\newcommand{\SAMP}{\mathsf{SAMP}}
\newcommand{\Ber}{\operatorname{Ber}}

\newcommand{\Poi}{\operatorname{Poi}}

\newcommand{\F}{\mathcal{F}}

\newcommand{\Test}{\textsc{NonadaptiveEquivalence}}

\renewcommand{\epsilon}{\varepsilon}

\newcommand{\cI}{\mathcal{I}}

\title{Tight Bounds for Equivalence Testing \\ with Non-Adaptive Conditional Samples}
\author{Gautam Kamath\footnote{New York University. {\tt g@csail.mit.edu}. Supported by an NSERC Discovery Grant and an Ontario Early Researcher Award.}}

\begin{document}
\maketitle

\begin{abstract}
	We study distribution testing with access to non-adaptive conditional samples.
	Specifically, we give tight bounds for equivalence testing, determining whether two unknown distributions are equal to or $\varepsilon$-far from each other in total variation distance. 
	Our algorithm and lower bound show that $\tilde \Theta\left(\frac{\log n}{\varepsilon^2}\right)$ queries are necessary and sufficient for this problem. 
	These results demonstrate that the complexity of uniformity, identity, and equivalence testing with non-adaptive conditional samples are all $\tilde \Theta(\log n)$. 
\end{abstract}

\section{Introduction}

Hypothesis testing is a classic problem in statistics: given samples from a distribution, determine whether or not it satisfies some property.
Within the field of theoretical computer science, it has been studied under the name of \emph{distribution testing}, as a subproblem within the broader area of \emph{property testing}~\cite{Goldreich17}. 
Study of distribution testing was initiated by seminal works of Goldreich, Goldwasser, and Ron~\cite{GoldreichGR96}, Goldreich and Ron~\cite{GoldreichR00}, and Batu, Fortnow, Rubinfeld, Smith, and White~\cite{BatuFRSW00}. 
The problem has since received extensive attention (see, e.g.,~\cite{BatuFFKRW01,BatuKR04,Paninski08,AcharyaDJOP11,BhattacharyyaFRV11,Valiant11,ValiantV11a,IndykLR12,DaskalakisDSVV13,ChanDVV14,ValiantV17a,Waggoner15,AcharyaD15,AcharyaDK15,BhattacharyaV15,DiakonikolasKN15a,DiakonikolasK16,Canonne16,BlaisCG17,BatuC17,CanonneDKS17,DaskalakisDK18,DaskalakisKW18,DiakonikolasGPP18,DiakonikolasGKPP21,CanonneJKL22}, and~\cite{Rubinfeld12,BalakrishnanW18,Kamath18, Canonne20b, Canonne22} for surveys).

We now have a rich understanding of many core problems in this area.
For example, the most classic question is \emph{uniformity testing}: given samples from a discrete distribution $p$ over $[n]$, test whether it is uniform over the domain versus $\varepsilon$-far from the uniform distribution in statistical distance. 
We have tight upper and lower bounds for this problem, demonstrating that the sample complexity of uniformity testing is $\Theta(\sqrt{n}/\varepsilon^2)$~\cite{GoldreichR00, Paninski08, ValiantV14, DiakonikolasKN15a, AcharyaDK15, DiakonikolasGPP18}. 
In this work, we are concerned with \emph{equivalence testing} (sometimes called \emph{two-sample testing} or \emph{closeness testing}): given samples from two discrete distributions $p$ and $q$ over $[n]$, test whether they are equal versus $\varepsilon$-far from each other in statistical distance. 
The sample complexity for this problem is the more complex expression $\Theta\left(\max\left\{\frac{\sqrt{n}}{\varepsilon^2}, \frac{n^{2/3}}{\varepsilon^{4/3}}\right\}\right)$~\cite{Valiant11, ChanDVV14}.

Observe that the sample complexity for both these problems (and many others in distribution testing) is (strongly) sublinear in the domain size $n$: ignoring dependence on $\varepsilon$, they are both $\Theta(n^c)$ for some $0 < c < 1$. 
However, these bounds are tight, which can be problematic when the domain size $n$ is extremely large and thus any polynomial in $n$ is intractable. 
For example, testing uniformity on the $d$-dimensional Boolean hypercube requires $\Theta(2^{d/2})$ samples, which remains exponential in the dimension.
Motivated by testing over extremely large domains, Canonne, Ron, and Servedio~\cite{CanonneRS14, CanonneRS15} and Chakraborty, Fischer, Goldhirsh, and Matsliah~\cite{ChakrabortyFGM13, ChakrabortyFGM16} introduced the \emph{conditional sampling model}. 
This model gives the algorithm stronger query access to the distributions being tested: rather than passively sampling from the distribution, the algorithm \emph{queries} the distribution by specifying a query set $S \subseteq [n]$ and receiving a sample from the distribution \emph{conditional on being from $S$}. 
We use $\COND$ to denote this type of conditional sampling access; in contrast, we use $\SAMP$ to denote the standard sampling oracle.

$\COND$  can significantly reduce the complexity of many distribution testing problems.
For example, uniformity testing (and the more general problem of \emph{identity testing}, which tests whether $p$ is equal to a specific distribution on $[n]$) requires only $\tilde \Theta(1/\varepsilon^2)$ conditional samples~\cite{CanonneRS14,FalahatgarJOPS15}, entirely independent of the domain size, compared to $\Theta(\sqrt{n}/\varepsilon^2)$ samples under $\SAMP$. 
Equivalence testing also enjoys dramatic savings: under $\COND$, equivalence testing needs only $\tilde \Theta(\log \log n)$ samples~\cite{FalahatgarJOPS15,AcharyaCK15b,ChakrabortyCK24,AdarFL26}, doubly logarithmic in the domain size, compared to $\Theta\left(\max\left\{\frac{\sqrt{n}}{\varepsilon^2}, \frac{n^{2/3}}{\varepsilon^{4/3}}\right\}\right)$ samples under $\SAMP$.
Motivated by the power of $\COND$, there has been significant investigation into its implications on distribution testing~\cite{Canonne15,FalahatgarJOPS15,AcharyaCK15b,RonT16,FischerLV17,SardharwallaSJ17,BlaisCG17,KamathT19,Narayanan21,ChakrabortyKM23,ChakrabortyCKM24,AdarFL26,BiswasBCS26}, in particular under the more restrictive subcube conditioning model~\cite{BhattacharyyaC18,CanonneCKLW21,ChenJLW21,BlancaCSV23, AdarFL24, ChenM24, MeelKP25,ChakrabartyCRSW25,GayHKO26,LeeNQR26}, as well as for adjacent problems such as group testing~\cite{AcharyaCK15a}, sublinear algorithms~\cite{GouleakisTZ17}, and crowdsourcing~\cite{GouleakisTZ18}.

However, most algorithms in the conditional sampling model are \emph{adaptive}: they assume that the results of a query can be observed before the algorithm specifies the next query set. 
One can thus naturally consider a \emph{non-adaptive} variant, where the algorithm must specify all query sets in advance of observing their results. 
We use $\NACOND$ to denote the non-adaptive conditional sampling model. 
What is the difference between these models -- how powerful is adaptivity? 
Prior work by Acharya, Canonne, and Kamath~\cite{AcharyaCK15b} showed that $\NACOND$ can be \emph{exponentially weaker} than $\COND$: they showed that even testing uniformity under $\NACOND$ requires $\Omega(\log n)$ samples, whereas the harder problem of testing equivalence requires only $O(\log \log n)$ samples~\cite{FalahatgarJOPS15}. 
However, later work of Kamath and Tzamos showed that $\NACOND$ remains \emph{exponentially stronger} than $\SAMP$: they give $O(\operatorname{poly}\log n)$ sample algorithms for uniformity, identity, and equivalence testing under $\NACOND$, whereas even uniformity testing under $\SAMP$ requires $\Omega(\sqrt{n})$ samples~\cite{Paninski08}. 

Thus, prior work has established exponential (or greater) separations between $\SAMP$, $\COND$, and $\NACOND$ for these testing problems.
However, a significant algorithmic gap was left by Kamath and Tzamos~\cite{KamathT19}.
Under $\NACOND$, they gave an $\tilde O(\log n)$-sample algorithm for uniformity testing, which nearly matches the $\Omega(\log n)$ lower bound of Acharya, Canonne, and Kamath~\cite{AcharyaCK15b}. 
However, their algorithms for identity and equivalence testing were significantly worse, at $O(\log^2 n)$ and $O(\log^{12} n)$ samples respectively, significantly greater than the $\Omega(\log n)$ lower bound (inherited from the easier problem of uniformity).
This quantitative gap is in sharp contrast to both $\SAMP$ and $\COND$, where we now have an essentially tight understanding of all three testing problems. 
A recent work of Chakraborty, Chakraborty, Kumar, and Meel (CCKM)~\cite{ChakrabortyCKM24} provided an algorithm for equivalence testing that requires $\tilde O(\log n)$ $\COND$ samples with \emph{only one round} of adaptivity.

\section{Results and Discussion}

We settle the complexity of equivalence testing with non-adaptive conditional samples. 
Specifically, we show that $\tilde \Theta\left(\frac{\log n}{\varepsilon^2}\right)$ samples are both necessary and sufficient for equivalence testing. 

\begin{theorem}
	\label{thm:equiv-main}
	Let $0 < \varepsilon \leq 1/2$, and let $p$ and $q$ be two (unknown) distributions over $[n]$. 
	There exists an algorithm that takes $\tilde O \left(\frac{\log n}{\varepsilon^2}\right)$ non-adaptive conditional samples from $p$ and $q$ and, with probability $\geq 2/3$, distinguishes between the case $p = q$ and $\dtv(p,q) \geq \varepsilon$. 
	Furthermore, any algorithm with such guarantees requires at least $\Omega\left(\frac{\log n}{\varepsilon^2}\right)$ non-adaptive conditional samples. 
\end{theorem}

Results are summarized in Table~\ref{tab:testing-complexities}. 

\begin{table}[H]
\centering
\footnotesize
\setlength{\tabcolsep}{3pt}
\renewcommand{\arraystretch}{1.15}
\begin{tabular}{@{}lccc@{}}
\toprule
Model & Uniformity & Identity & Equivalence \\
\midrule
$\SAMP$
& $\Theta\!\left(\frac{\sqrt n}{\varepsilon^2}\right)$~\cite{Paninski08,ValiantV14}
& $\Theta\!\left(\frac{\sqrt n}{\varepsilon^2}\right)$~\cite{Paninski08,ValiantV14}
& $\Theta\!\left(\max\!\left\{\frac{n^{2/3}}{\varepsilon^{4/3}},\frac{\sqrt n}{\varepsilon^2}\right\}\right)$~\cite{ChanDVV14} \\
\midrule
$\NACOND$
& $\tilde O\!\left(\frac{\log n}{\varepsilon^2}\right)$~\cite{KamathT19}
& $\tilde O\!\left(\frac{\log n}{\varepsilon^2}\right)$~\textbf{This work}
& $\tilde O\!\left(\frac{\log n}{\varepsilon^2}\right)$~\textbf{This work} \\
& $\Omega\!\left(\log n+\frac{1}{\varepsilon^2}\right)$~\cite{AcharyaCK15b}
& $\Omega\!\left(\log n+\frac{1}{\varepsilon^2}\right)$~\cite{AcharyaCK15b}
& $\Omega\!\left(\frac{\log n}{\varepsilon^2}\right)$~\textbf{This work} \\
\midrule
$\COND$
& $\tilde O\!\left(\frac{1}{\varepsilon^2}\right)$~\cite{CanonneRS15}
& $\tilde O\!\left(\frac{1}{\varepsilon^2}\right)$~\cite{FalahatgarJOPS15}
& $\tilde O\!\left(\frac{\log\log n}{\varepsilon}+\frac{1}{\varepsilon^5}\right)$~\cite{AdarFL26} \\
& $\Omega\!\left(\frac{1}{\varepsilon^2}\right)$
& $\Omega\!\left(\frac{1}{\varepsilon^2}\right)$
& $\tilde\Omega\!\left(\log\log n+\frac{1}{\varepsilon^2}\right)$~\cite{ChakrabortyCK24} \\
\bottomrule
\end{tabular}
\caption{Comparison of uniformity, identity, and equivalence testing in different sampling models. Problems get more challenging as one moves up and to the right.}
\label{tab:testing-complexities}
\end{table}

Our results address open questions that were explicitly raised by Kamath and Tzamos~\cite{KamathT19}. 
For equivalence testing under $\NACOND$, we improve upon the sample complexity of their algorithm, from $\tilde O\left(\frac{\log^{12} n}{\varepsilon^2}\right)$ to the near-optimal $\tilde O\left(\frac{\log n}{\varepsilon^2}\right)$.
We also improve the lower bound from $\Omega(\log n)$ samples (inherited from the easier problem of uniformity testing~\cite{AcharyaCK15b}) to the nearly-matching $\Omega\left(\frac{\log n}{\varepsilon^2}\right)$ samples. 
Since identity testing is easier than equivalence testing, the same algorithm improves the $\tilde O\left(\frac{\log^2 n}{\varepsilon^2}\right)$ sample complexity of Kamath and Tzamos~\cite{KamathT19} to $\tilde O\left(\frac{\log n}{\varepsilon^2}\right)$.
These results show that, under $\NACOND$, the sample complexities of testing uniformity, identity, and equivalence are all $\tilde \Theta(\log n)$. 
We find this somewhat surprising, since identity and equivalence testing under $\SAMP$ and $\COND$ have significant gaps: under $\SAMP$, the sample complexity increases by a $\operatorname{poly}(n)$ factor, while under $\COND$, an entirely new dependence on $n$ arises. 

We highlight one interesting question, on the power of rounds of adaptivity of conditional samples. 
Falahatgar, Jafarpour, Orlitsky, Pichapati, and Suresh (FJOPS)~\cite{FalahatgarJOPS15} give an identity testing algorithm that takes $\tilde O(1/\varepsilon^2)$ conditional samples, nearly matching the lower bound of $\Omega(1/\varepsilon^2)$. 
A careful study of their algorithm shows that it can be implemented with only \emph{one round of adaptivity}. 
In contrast, any $\NACOND$ algorithm (which has zero rounds of adaptivity) requires $\tilde \Theta(\log n)$ samples. 
Thus, \emph{a single round of adaptivity} extracts the full power of adaptive conditional samples.
In contrast, for equivalence testing, it is unclear how many rounds of adaptivity are required to reduce the sample complexity from $\tilde O(\log n)$ (optimal under $\NACOND$) to $\tilde O(\log \log n)$ (optimal under $\COND$).

\subsection{Techniques}

\paragraph{Upper Bound.}
Our algorithm (Algorithm~\ref{alg:main}) is very similar to \textsc{Anaconda} of Kamath and Tzamos~\cite{KamathT19}. 
Their algorithm operates by repeatedly isolating a query set $S$ by first choosing a random probability $1/r$, and then independently including each element in $S$ with probability $1/r$. 
If $\dtv(p,q)\geq \varepsilon$, they argue that this approach has a reasonable chance of identifying a set $S$ which is discrepant between the two distributions: specifically, there exists a symbol $i \in S$ such that the (conditional) probability of $i$ has a large absolute difference between $p$ and $q$.
Given such a set $S$, we can easily test for such a discrepancy by examining the empirical distribution. 

Our algorithm diverges from theirs in two primary ways, incorporating ideas from the \emph{adaptive} algorithm for equivalence testing with conditional samples of FJOPS~\cite{FalahatgarJOPS15}. 
Specifically, we replace \textsc{Anaconda}'s discrepancy measure with a chi-squared-like discrepancy measure of FJOPS. 
This discrepancy captures \emph{relative} differences, rather than absolute differences, which can be valuable when the symbols have small mass. 
Additionally, our algorithm incorporates another loop which searches over the magnitude of the discrepancy that we test for.
Intuitively, this can be seen as a better allocation of our testing effort across scales -- if we are searching for a large discrepancy (which requires few samples to test), we take proportionally more random trials in an attempt to capture the discrepancy. 
This tradeoff will show up in the final work investment part of the proof. 

Analysis is inspired by the \emph{adaptive} equivalence testing algorithm of FJOPS. 
Consider the domain to be sorted in decreasing order with respect to the distribution $w = \frac{p+q}{2}$ (``the weight''). 
Let $T_i = \sum_{\ell \geq i} w_\ell$ be the aggregate mass in the tail beginning at symbol $i$. 

FJOPS randomly selects several symbols according to the distribution $w$, in hopes that if $\dtv(p,q) \geq \varepsilon$, one will be discrepant relative to its tail (quantified by the expectation with respect to $w$ in Lemma~\ref{bb:tail}).
For one such symbol $i$, they try to estimate $\rho_i = w_i/T_i$.
Their goal is to form a query set $S$ that simultaneously a) has $i$ be the highest weight symbol selected, and b) have the weight $w_i$ be comparable to the selected weight from tail elements. 
This can be achieved with reasonable probability by selecting each domain element with probability $\rho_i$.
With such a set $S$ in hand, the discrepancy test of Corollary~\ref{cor:test-discrep} allows us to identify that this is the case. 

In our non-adaptive setting, a key challenge is that we are unable to estimate $\rho_i$.
Instead, the argument goes in the opposite order: we randomly select a $\rho$, and for each symbol $i$ such that $\rho \approx w_i$, we consider the event a), that $i$ is the highest weight symbol selected.
We can argue that, under this event, the symbol $i$ contributes significant expected discrepancy (Lemma~\ref{lem:coordinate}).
This time, we do not require b), that the weight of symbol $i$ and the sampled tail are comparable -- Lemma~\ref{lem:quad-tail} supports a different approach, where we examine the expected discrepancy of the random tail. 
We can sum these contributions across all $i$'s (since the corresponding events are disjoint) and achieve a large expected discrepancy in aggregate (Corollary~\ref{cor:expected-discrep}). 

A large expected discrepancy does not suffice for our algorithm: we require that the algorithm actually obtains a query set with significant discrepancy.
We show that, given a large expected discrepancy, our algorithm is likely to sample a sufficiently discrepant set  (quantified in Lemma~\ref{lem:energy}). 
This conversion from expectation to probability is done using Levin's economical work-investment strategy. 
Again, with a query set $S$ in hand, we can identify a discrepancy using the test of Corollary~\ref{cor:test-discrep}.

\paragraph{Lower Bound.}
To construct a hard instance for our lower bound, we describe two distributions $p$ and $q$. 
Both distributions share the same support $D$. 
The support size $|D|$ is chosen randomly from $\{2^1, \dots, 2^{\log_2 n}\}$, and the specific support $D \subseteq [n]$ is chosen uniformly at random from all sets of that size. 
This construction is similar to that of Acharya, Canonne, and Kamath~\cite{AcharyaCK15b}, who used it to prove lower bounds for non-adaptive support size estimation and uniformity testing. 
For the case when $p = q$, we let both distributions be uniform over this support. 
For the case where $\dtv(p,q) \geq \varepsilon$, we apply a Paninski-style construction~\cite{Paninski08} to the support. 
Specifically, we randomly partition $D$ into two sets of size $|D|/2$: we multiplicatively perturb the probability of elements in these sets by $1+\varepsilon$ and $1-\varepsilon$, respectively, for the distribution $p$. 
Keeping the same partition, we perturb the probabilities in the opposite direction ($1-\varepsilon$ and $1+\varepsilon$) for the distribution $q$. 
The two types of randomization each contribute a different factor to the lower bound.
The randomization over the support size contributes a factor of $\log n$, while the random perturbation contributes a factor of $1/\varepsilon^2$.

The proof works by bounding the total variation distance between the transcripts (i.e., the results of the conditional samples) when $p=q$ and $\dtv(p,q) \geq \varepsilon$. 
We start by focusing on the case when the support $D$ and queries are fixed, and the randomness is only in the choice of the random perturbation (Section~\ref{sec:fixed-support}).
We follow the second-moment method, and bound the total variation distance in terms of the expected square of the likelihood ratio of the transcripts. 
Conditioning on a fixed perturbation, the likelihood ratio of the transcript can be written as a product of the likelihood ratio of individual results in the transcript.
However, in the case where $\dtv(p,q) \geq \varepsilon$, we average over a single perturbation shared by all queries, which introduces correlations when conditioning is removed. 
Controlling these correlations is the main challenge in our argument. 
A lengthy calculation allows us to bound the total variation distance between the transcripts as $O\left(\varepsilon^2 \sum_i h(|S_i \cap D|)\right)$, where $h$ is a function solely of the size of the intersection between the query set $S_i$ and the support $D$ (Lemma~\ref{lb:lem:transcript}). 
Along the way, we require a variant of the Khintchine inequality (Lemma~\ref{lb:lem:signs}): rather than taking a weighted sum with uniformly random signs, the random partition of $D$ introduces a balancedness condition which we must account for. 

Next, we incorporate randomness over the choice of the support $D$. 
As hinted in Lemma~\ref{lb:lem:transcript}, the key quantity is the size of the intersection between a query set $S$ and the support $D$: $|S \cap D|$. 
If the intersection is small (i.e., $0$ or $1$), the distribution of the query result will be identical under $p=q$ versus $\dtv(p,q)\geq \varepsilon$. 
On the other hand, if the intersection is large, averaging over the random partition results in a relatively low contribution to the distance between the transcripts.
Indeed, the only queries that incur significant distance are those where $|S \cap D| \approx \Theta(1)$. 
An averaging argument over the size of the set $D$ bounds the expected value of $h(|S \cap D|)$ by $O(1/\log n)$ (Lemma~\ref{lb:lem:scales}).

A careful combination of both sources of randomness (as well as incorporation of the random seed of the algorithm) allows us to bound the total variation distance between two transcripts over $N$ samples as $O(\varepsilon^2 N/\log n)$. 
Accuracy of the testing algorithm requires the distance between the transcripts to be $\Omega(1)$, thus necessitating $N = \Omega(\log n/\varepsilon^2)$, as desired. 

\subsection*{AI Disclosure}
The algorithm and analysis were developed in collaboration with an AI model. 
The model initially failed to attribute ideas to other work in the literature, identify lemmas as variants of known results, and write the solution in a manner easily comprehensible to a human reader. 
To the best of our ability, we attempted to remedy all of these issues, including carefully writing all parts of the paper. 
We take full responsibility for correctness, presentation, and communication of the results. 

\section{Preliminaries}
We study discrete distributions over $[n]$.
We denote the (unknown) distributions of interest using $p$ and $q$, where $p(i)$ is the probability mass placed by distribution $p$ on symbol $i$ (and similarly for $q$). 
For a set $S\subseteq[n]$, let $p(S)=\sum_{i\in S}p(i)$. 
If $p(S)>0$, write $p_S$ for the distribution obtained by conditioning $p$ on $S$:
\[
 p_S(i)=\begin{cases}p(i)/p(S),&i\in S,\\0,&i\notin S.\end{cases}
\]

\begin{definition}[Conditional sampling~\cite{CanonneRS14,ChakrabortyFGM13}]\label{def:oracle}
A conditional sampling oracle for $p$ takes a query set $S\subseteq[n]$ and returns an independent sample from $p_S$. 
If $p(S)=0$, then a symbol $i \in S$ is returned uniformly at random.
\end{definition}

We focus on non-adaptive conditional access ($\NACOND$), in which all query sets must be chosen before observing any of the responses. 
Alternatively, one could consider adaptive conditional access ($\COND$), in which query sets may depend on previous responses. 

\begin{definition}
	The \emph{total variation distance} between distributions $p$ and $q$ is defined as 
	\[\dtv(p,q) = \frac{1}{2} \sum_{i=1}^n |p(i) - q(i)|.\]
\end{definition}

\subsection{Useful Tools from Previous Work}
We borrow two tools from the work of Falahatgar, Jafarpour, Orlitsky, Pichapati, and Suresh~\cite{FalahatgarJOPS15}.

\paragraph{A chi-squared-like discrepancy test.}
Prior algorithms on non-adaptive conditional sampling~\cite{KamathT19} worked as follows: repeatedly draw a random query set and test whether we can identify a discrepancy within the conditional distributions. 
Kamath and Tzamos searched for discrepancies based on absolute differences between conditional symbol probabilities, i.e., symbols for which $|p_S(i) - q_S(i)|$ are sufficiently large. 
We follow the same general framework, but we search for discrepancies based on \emph{relative} differences between conditional symbol probabilities, as proposed by FJOPS~\cite{FalahatgarJOPS15}. 
We employ their algorithm \textsc{Test-equal}: 

\begin{lemma}[Lemma 1 of~\cite{FalahatgarJOPS15}]
	\label{lem:test-equal}
	Fix $0 < t, \delta \leq 1$. 
	Suppose we have samples from two distributions $\Ber(u)$ and $\Ber(v)$, where $u$ and $v$ are unknown.
	There exists an algorithm, \textsc{Test-equal}, that, with probability $\geq 1 - \delta$, distinguishes between the cases where $u = v$ and 
	\[ \frac{(u-v)^2}{(u+v)(2-u-v)} \geq t.\]
	It requires $O\left(\frac{\log (1/\delta)}{t}\right)$ samples from each distribution. 
\end{lemma}

We reproduce the pseudocode of \textsc{Test-equal} from~\cite{FalahatgarJOPS15} as Algorithm~\ref{alg:test-equal} in Appendix~\ref{sec:test-equal}.

The above lemma is a test for discrepancy between Bernoulli distributions. 
The following corollary allows us to test whether two discrete distributions have a discrepant coordinate. 
The reduction to Lemma~\ref{lem:test-equal} is simple: given a batch of samples, define an indicator for each unique symbol observed, and run the previous Bernoulli test (with a scaled-down failure probability) for each of these events. 
The proof is detailed further in Appendix~\ref{sec:test-discrep}.
\begin{corollary}
	\label{cor:test-discrep}
	Fix $0 < t, \delta \leq 1$. 
	Suppose we have samples from two discrete distributions $p$ and $q$ over $[n]$.
	There exists an algorithm that, with probability $\geq 1 - \delta$, distinguishes between the cases where $p = q$ and there exists some $i \in [n]$ such that 
	\[ \frac{(p(i)-q(i))^2}{(p(i)+q(i))(2-p(i)-q(i))} \geq t.\]
	It requires $O\left(\frac{\log(1/t) + \log (1/\delta)}{t}\right)$ samples from each distribution. 

\end{corollary}

\paragraph{A Tail-Discrepancy Lemma.}

Let $w, v$ be $n$-dimensional vectors defined as follows:\footnote{Throughout the paper, we use the convention that $0/0 = 0$.}
\begin{equation*}
	w_i =\frac{p(i)+q(i)}2,\quad
	v_i=\frac{p(i)-q(i)}{p(i)+q(i)}.\quad
\end{equation*}
$w$ is a distribution which is a uniform mixture of $p$ and $q$, while $v$ is a rescaled difference between distributions $p$ and $q$. 
Observe that $\sum_{i=1}^n w_i v_i = 0$ and $\sum_{i=1}^n w_i |v_i| = \dtv(p,q)$. 

In the analysis of the algorithm, we assume that $w_1 \geq \dots \geq w_n \geq 0$.
Since this ordering will only be used for the sake of analysis (not by the algorithm itself), this is without loss of generality. 
We further define:
\begin{equation*}
	T_i =\sum_{\ell=i}^n w_\ell,\quad
	m_i=\frac{\sum_{\ell=i}^n w_\ell v_\ell}{T_i}.
\end{equation*}

$T$ and $m$ are vectors which quantify tail properties of $p$ and $q$. 
Specifically, $T$ represents the tail mass of their mixture $w$, and $m$ measures the (rescaled) difference between $p$ and $q$ in this tail. 


The following structural lemma of~\cite{FalahatgarJOPS15} allows us to relate these quantities to the total variation distance between $p$ and $q$.
Specifically, if $p$ and $q$ are far from each other, this manifests in the differences between their local difference and their tail difference (appropriately rescaled).

\begin{lemma}[Lemma 8 of~\cite{FalahatgarJOPS15}]\label{bb:tail}
\begin{equation*}
\sum_{i=1}^{n}w_i|v_i-m_i|
\geq\tfrac12\TV(p,q).
\end{equation*}
\end{lemma}

\section{The Algorithm}\label{sec:algorithm}

Our algorithm is described as Algorithm~\ref{alg:main}. 
It is very simple to describe: informally, it chooses a random set of a random size, and checks whether the two distributions are noticeably discrepant on the induced conditional distributions. 
This is independently repeated several times, and further wrapped in a loop which iterates over choices for the magnitude of the discrepancy. 
Crucially, note that this algorithm is non-adaptive, as all queries can be performed oblivious to the results of previous queries. 

We note that our algorithm is almost identical to \textsc{Anaconda}~\cite{KamathT19}. 
There are two main differences. 
Most significantly, we search for relative discrepancy (as described in Corollary~\ref{cor:test-discrep}), rather than absolute discrepancy. 
We also include a new outer loop which iterates over possibilities for the size of the discrepancy. 
Lower order differences with \textsc{Anaconda} include various parameter settings. 

The pseudocode uses the following parameter settings for the number of repetitions $T(t)$ and the number of queries $m(t)$. 
Note that both quantities depend on the candidate magnitude for the discrepancy $t$. 
\[
\begin{aligned}
T(t)&=O\!\left(\frac{t}{\varepsilon^2}\log\frac{n}{\varepsilon}\log\frac1\varepsilon\right),\\
m(t)&=O\!\left(\frac{\log\log(n/\varepsilon)+\log(1/\varepsilon)}{t}\right).
\end{aligned}
\]
Each discrepancy test has failure probability $\delta$, set small enough to bound failure across all tests by a constant. 
\[
\delta=O\!\left(\frac{\varepsilon^2}{\log(n/\varepsilon)\log(1/\varepsilon)}\right).
\]

\begin{algorithm}[H]
\caption{An algorithm for equivalence testing with non-adaptive conditional samples}
\label{alg:main}
\begin{algorithmic}[1]
\Function{\Test}{$n,\varepsilon,\NACOND_p,\NACOND_q$}
\For{each threshold $t=1,1/2,1/4,\ldots,\Theta\!\left(\varepsilon^2/(\log(n/\varepsilon)\log(1/\varepsilon))\right)$} \label{ln:thresholds}
	\For{$\ell=1,\ldots,T(t)$} \label{ln:repetitions}
	\State Choose an integer $j\in\{1,\ldots,O(\log(n/\varepsilon))\}$ uniformly at random, and define $r=2^j$. \label{ln:rand-prob}
    \State Choose a random set $S\subseteq[n]$, independently selecting each $i$ with probability $1/r$.
    \State Perform $m(t)$ queries to $\NACOND_p$ and $\NACOND_q$ on $S$.
    \State Use these samples to run the test of \cref{cor:test-discrep} with threshold $t$, failure prob.\ $\delta$.
	\If{the discrepancy test returns \texttt{diff}}
      \State \Return $\dtv(p,q)\geq\varepsilon$.
    \EndIf
  \EndFor
\EndFor
\State \Return $p=q$.
\EndFunction
\end{algorithmic}
\end{algorithm}

There are $O(\log\log(n/\varepsilon)+\log(1/\varepsilon))$ thresholds. The sample complexity of this algorithm is 
\begin{align*}
	\sum_t T(t) m(t) &= \sum_t O\left(\frac{t}{\varepsilon^2}\log\frac{n}{\varepsilon}\log\frac1\varepsilon\right)\cdot O\left( \frac{\log\log(n/\varepsilon)+\log(1/\varepsilon)}{t}\right) \\
	&= \sum_t  O\left(\frac{\log (n/\varepsilon) \log (1/\varepsilon) \log \log (n/\varepsilon) + \log (n/\varepsilon) \log^2 (1/\varepsilon)}{\varepsilon^2}\right) \\
	&=  O\left(\frac{\log (n/\varepsilon)\log (1/\varepsilon)\left(\log\log(n/\varepsilon)+\log(1/\varepsilon)\right)^2}{\varepsilon^2}\right) \\
	&=  \tilde O\left(\frac{\log n}{\varepsilon^2}\right).
\end{align*}

\section{Analysis of Algorithm~\ref{alg:main}}\label{sec:analysis}

For a query set $S \subseteq [n]$, we define $\F(S)$ as follows: 
\[
\F(S)=\max_{i\in S}
\frac{(p_S(i)-q_S(i))^2}
{(p_S(i)+q_S(i))(2-p_S(i)-q_S(i))}.
\]
This quantity identifies the maximum discrepancy (in the sense of Corollary~\ref{cor:test-discrep}) of any symbol in the conditional distribution corresponding to $S$. 
The following key lemma says that, over the randomness in sampling the various query sets $S$, Algorithm~\ref{alg:main} is likely to sample at least one set that realizes a large discrepancy. 

\begin{lemma}\label{lem:energy}
Suppose $\TV(p,q)\geq\varepsilon$. 
We let $S_{t,\ell}$ be the query set sampled by Algorithm~\ref{alg:main} at repetition $\ell$ with threshold $t$.
Then with probability at least $9/10$, there exists $t, \ell$ such that $\F(S_{t,\ell})\geq t$.
\end{lemma}

With this lemma in hand, the proof of the upper bound in Theorem~\ref{thm:equiv-main} is not hard. 
If $p=q$, then for all the query sets $S$, we have that $p_S = q_S$. 
If $\dtv(p,q) \geq \varepsilon$, this lemma guarantees that we will identify a query set $S_{t,\ell}$ realizing discrepancy at least $t$ with probability at least $9/10$. 
Correctness follows from the guarantees of the discrepancy test in Corollary~\ref{cor:test-discrep} and a union bound over all $\sum_t T(t) = O\left(\frac{\log (n/\varepsilon) \log (1/\varepsilon)}{\varepsilon^2}\right)$ invocations.
The claimed sample complexity follows from the calculation in Section~\ref{sec:algorithm}.

It remains to prove Lemma~\ref{lem:energy}. 
As a first step, we convert the tail-discrepancy lemma of FJOPS~\cite{FalahatgarJOPS15} (Lemma~\ref{bb:tail}) to a quadratic variant which will be more convenient for downstream calculations (Lemma~\ref{lem:quad-tail} in Section~\ref{sec:quad-tail}). 
As our main technical lemma, we show that a randomly sampled set will have large discrepancy in expectation (Corollary~\ref{cor:expected-discrep} in Section~\ref{sec:structural}). 
An application of Levin's work-investment strategy (Section~\ref{sec:investment}) allows us to convert this expectation into getting a set having sufficient discrepancy with sufficient probability, concluding Lemma~\ref{lem:energy}. 

\subsection{A Quadratic Tail-Discrepancy Lemma}
\label{sec:quad-tail}

We show that the linear tail-discrepancy lemma of FJOPS~\cite{FalahatgarJOPS15} can be converted to a quadratic version, using Cauchy-Schwarz inequality. 
It turns out that these quadratic differences are better suited for the chi-square-like discrepancy we are trying to witness with Corollary~\ref{cor:test-discrep}.
At the same time, we exclude symbols $i$ which have extremely small weight $w_i$ under $p$ and $q$. 
Later analysis derives lower bounds on discrepancy exclusively from when a symbol is the heaviest element sampled and its weight is comparable to the sampling probability (see the definition of event $A_i$ and Lemma~\ref{lem:coordinate} in Section~\ref{sec:structural}).
Since our algorithm has a minimum sampling probability of $\Theta\left(\varepsilon/n\right)$, any item with less weight will never satisfy the latter criterion.
Consequently, to match the events we are able to extract discrepancy from, we aggregate over the corresponding set of symbols. 
At the same time, ignoring the remaining symbols loses only a small amount of the witnessed total variation distance.

\begin{lemma}\label{lem:quad-tail}
Define the following set of indices: 
\[
\cI=\{i: w_i \geq \varepsilon/16n,\ T_i\geq \varepsilon/16\}.
\]
If $\TV(p,q)\geq\varepsilon$, then
\begin{equation*}
\sum_{i\in\cI}w_iT_i(v_i-m_i)^2
	\geq\frac{\varepsilon^2}{16\left(1+\ln\frac{16}{\varepsilon}\right)}.
\end{equation*}
\end{lemma}
\begin{proof}
	The coordinates with $w_i \leq \varepsilon/16n$ have total probability mass at most $n \cdot \varepsilon/16n = \varepsilon/16$. 
	Similarly, the coordinates with $T_i \leq \varepsilon/16$ have total mass at most $\varepsilon/16$.
	Observing that $v_i, m_i \in [-1,1]$ entrywise, we have that $|v_i - m_i| \leq 2$. 
	Thus, removing $[n] \setminus \cI$ from the summation in \cref{bb:tail} gives 
\begin{equation*}
\sum_{i\in\cI}w_i|v_i-m_i|
\geq\varepsilon/2-2(\varepsilon/16+\varepsilon/16)=\varepsilon/4.
\end{equation*}
Thus, 
\begin{equation}\label{eq:retained-l1}
	\left(\sum_{i\in\cI}w_i|v_i-m_i|\right)^2 \geq \varepsilon^2/16.
\end{equation}
We are left with $\cI = [r]$ for some $r \leq n$. 
	Since $w_i=T_i-T_{i+1}$,
	\begin{equation}\label{eq:telescoping}
\sum_{i=1}^r\frac{w_i}{T_i}
	= \sum_{i=1}^r\frac{T_i-T_{i+1}}{T_i} = \sum_{i=1}^r 1- \frac{T_{i+1}}{T_i}
\leq1+\sum_{i=1}^{r-1}\ln\frac{T_i}{T_{i+1}}
=1+\ln\prod_{i=1}^{r-1}\frac{T_i}{T_{i+1}}
=1+\ln\frac1{T_r}
	\leq1+\ln\frac{16}{\varepsilon}.
\end{equation}
The first inequality bounds the $i=r$ term by $1$, and the remaining terms using $1-x\leq - \ln x$. 

\noindent We factor the summand of the summation as follows: 
\[
w_i|v_i-m_i|=
\sqrt{w_iT_i(v_i-m_i)^2}\sqrt{w_i/T_i}.
\]
Cauchy-Schwarz says that 
\[
	\left(\sum_{i \in \cI} w_i|v_i-m_i|\right)^2\leq 
	\left(\sum_{i \in \cI} w_iT_i(v_i-m_i)^2\right)\left(\sum_{i \in \cI} w_i/T_i\right).
\]
Combining this with \eqref{eq:retained-l1} and \eqref{eq:telescoping} gives the desired statement. 
\end{proof}

\subsection{Expected Discrepancy under Random Restrictions}\label{sec:structural}

The goal of this section is to establish that the expected discrepancy $\E[\F(S)]$ is large, where the randomness is over the choice of the sampling probability and the selected elements of the query set $S$. 
In the following section, we demonstrate that our sampling strategy allows us to convert this large expectation into a reasonable probability of selecting a sufficiently discrepant set, establishing Lemma~\ref{lem:energy}.

We bucket the indices $i \in [n]$ based on their weight $w_i$. 
Specifically, we define 
\[ 
B_j=\{i:2^{-j}< w_i \leq 2^{-j+1}\}.
\]

We define $A_i$ to be the event where symbol $i$ is included in a query set $S$, and no earlier symbol $\ell < i$ is selected:
\[
A_i=\{i\in S,\ S\cap\{1,\ldots,i-1\}=\emptyset\}.
\]
The following proposition says that with a judicious choice of sampling probability, the probability that symbol $i$ is the heaviest selected symbol is comparable to the probability it is selected at all. 
\begin{proposition}
	\label{prop:i-heaviest}
	Fix $i \in [n]$, and let $j$ be the index of the bucket $B_j$ that $i$ belongs to.
	Let $\rho = 2^{-j}$ be a sampling probability, and let $S \subseteq [n]$ be a random set where every symbol in $[n]$ is independently included with probability $\rho$. 
	Then $\Prob_S(A_i \mid j) \geq \rho/4$. 
\end{proposition}
\begin{proof}
	First, we have that $A_i$ occurs precisely when symbol $i$ is selected and none of the previous $i-1$ symbols are selected: $\Prob_S(A_i \mid j) = \rho (1-\rho)^{i-1}$.
	Observe that $(i-1)w_i \leq iw_i \leq 1$: the latter inequality holds because the $w_i$'s are in decreasing order and form a probability distribution. 
	If this weren't the case, this prefix would contain more than $1$ mass. 
	Note further that $\rho \leq w_i$, based on the definition of our bucketing scheme. 
	Combining these inequalities, we have that $i-1 \leq 1/\rho$, and thus $\rho (1-\rho)^{i-1} \geq \rho (1-\rho)^{1/\rho} \geq \rho/4$ as desired, where the latter inequality holds for $\rho \leq 1/2$. 
\end{proof}




The following lemma captures the expectation of $\F(S)$ extracted when the event $A_i$ occurs: 

\begin{lemma}\label{lem:coordinate}
	Fix $i \in [n]$, and let $j$ be the index of the bucket $B_j$ that $i$ belongs to.
	Let $\rho = 2^{-j}$ be a sampling probability, and let $S \subseteq [n]$ be a random set where every symbol in $[n]$ is independently included with probability $\rho$. 
	Then 
\begin{equation*}\label{eq:coordinate}
\E_S[\F(S)\ind_{A_i} \mid j ]
\geq\frac{w_iT_i(v_i-m_i)^2}{2048},
\end{equation*}
where the randomness is only over the choice of the random set $S$. 
\end{lemma}

With this lemma in hand, we demonstrate how it allows us to conclude that the overall discrepancy we expect to capture is large.
The proof essentially aggregates over the $A_i$ events (which are disjoint), applies Lemma~\ref{lem:coordinate} to each, and then applies our quadratic tail-discrepancy lemma (Lemma~\ref{lem:quad-tail}) to their sum. 
We instantiate $\Theta(\log(n/\varepsilon))$ in Line~\ref{ln:rand-prob} of Algorithm~\ref{alg:main} as $L := 1 + \log_2\frac{16n}{\varepsilon}$.
\begin{corollary}
	\label{cor:expected-discrep}
Suppose $\TV(p,q)\geq\varepsilon$. Then
\[
	\E_{j,S}[\F(S)]\geq\frac{\varepsilon^2}{2^{15}\left(1+\log_2\frac{16n}{\varepsilon}\right)\left(1+\ln\frac{16}{\varepsilon}\right)},
\]
where the randomness is over the choice of the sampling probability and the selected elements of the query set $S$. 
\end{corollary}
\begin{proof}
We first use the law of total expectation over the random choice of $j$:
\[
\E_{j,S}[\F(S)]
=\frac1L\sum_{j=1}^{L}\E_S[\F(S)\mid j].
\]
Noting that the $A_i$'s are disjoint events, we have that 
\[
	\frac1L\sum_{j=1}^{L}\E_S[\F(S)\mid j] \geq \frac1L\sum_{j=1}^{L}\sum_{i \in B_j} \E_S[\F(S) \ind_{A_i}\mid j].
\]
Applying Lemma~\ref{lem:coordinate} gives 
\[
\frac1L\sum_{j=1}^{L}\sum_{i \in B_j} \E_S[\F(S) \ind_{A_i}\mid j]
	\geq\frac1{2048L}\sum_{j=1}^{L}\sum_{i \in B_j}w_iT_i(v_i-m_i)^2.
\]
Observing that $\cI \subseteq \cup_{j=1}^L B_j$,
\[
\frac1{2048L}\sum_{j=1}^{L}\sum_{i \in B_j}w_iT_i(v_i-m_i)^2 \geq 
\frac1{2048L}\sum_{i\in\cI}w_iT_i(v_i-m_i)^2.
\]
Applying Lemma~\ref{lem:quad-tail} and the assumption that $\dtv(p,q) \geq \varepsilon$, we can conclude:
\[
\E_{j,S}[\F(S)]
	\geq \frac{1}{2^{11}L} \cdot \frac{\varepsilon^2}{2^{4}\left(1+\ln\frac{16}{\varepsilon}\right)} =  \frac{\varepsilon^2}{2^{15}\left(1 + \log_2\frac{16n}{\varepsilon}\right)\left(1+\ln\frac{16}{\varepsilon}\right)}.
\]
\end{proof}

\noindent We proceed to prove Lemma~\ref{lem:coordinate}. 

\begin{proof}[Proof of \cref{lem:coordinate}]
	Fix coordinate $i \in [n]$, the index $j$ of the bucket $B_j$ it belongs to, and the sampling probability $\rho=2^{-j}$. 
	We will condition on the event $A_i$, where $i$ is the first symbol selected and every subsequent symbol is selected with probability $\rho$, and the choice of the index $j$. 
	Define the following quantities:
\begin{equation*}
U=\sum_{\ell\in S\setminus\{i\}}w_\ell,\qquad
W=w_i+U,\qquad
Y=\sum_{\ell\in S}w_\ell(v_i-v_\ell).
\end{equation*}
$U$ and $W$ are the total weight of the set $S$, without and with the contribution of symbol $i$. 
$Y$ captures how the relative difference between $p$ and $q$ at symbol $i$ compares to their relative difference (on average) across the rest of the set $S$.

We compute the (conditional) expectation $Y$: 
\begin{equation}\label{eq:EY}
\E_S[Y\mid A_i,j]=\rho\sum_{\ell\geq i}w_\ell(v_i-v_\ell)
=\rho T_i(v_i-m_i),
\end{equation}
	where we recall the definitions $T_i = \sum_{\ell=i}^n w_\ell$ and $m_i = \frac{\sum_{\ell=i}^n w_\ell v_\ell}{T_i}$. 

We establish the pointwise bound
\begin{equation}\label{eq:pointwise}
\F(S)\geq\frac{w_i^2Y^2}{2W^3U}.
\end{equation}
First, we deal with degenerate cases: if $U=0$, $p(S)=0$, or $q(S) = 0$, it is not hard to see that $Y=0$ and the bound holds vacuously.
Therefore, we assume that $p(S), q(S),$ and $U > 0$. 

We inspect the following lower bound for $\F(S)$, examining the contribution at symbol $i$. 
\[
\F(S)\geq\frac{(p_S(i)-q_S(i))^2}{(p_S(i)+q_S(i))(2-p_S(i)-q_S(i))}
\]
We first derive an alternative form for the numerator. 
We write
\[
p_S(i)-q_S(i)=\frac{p(i)q(S)-q(i)p(S)}{p(S)q(S)}.
\]
	Using $p(\ell) = \frac{p(\ell) + q(\ell)}{2}\cdot \frac{(p(\ell) + q(\ell)) + (p(\ell) - q(\ell))}{p(\ell) + q(\ell)}=w_\ell(1+v_\ell)$ and $q(\ell)=w_\ell(1-v_\ell)$, we expand the numerator as
\[
\begin{aligned}
p(i)q(S)-q(i)p(S)
&=\sum_{\ell\in S}\bigl(p(i)q(\ell)-q(i)p(\ell)\bigr)\\
&=w_i\sum_{\ell\in S}w_\ell
\bigl((1+v_i)(1-v_\ell)-(1-v_i)(1+v_\ell)\bigr)\\
&=2w_i\sum_{\ell\in S}w_\ell(v_i-v_\ell)
=2w_iY.
\end{aligned}
\]
Therefore,
\[
	(p_S(i)-q_S(i))^2
	=\frac{4w_i^2Y^2}{(p(S)q(S))^2}.
\]
To deal with the second term in the denominator, write each complement probability in terms of its remaining mass:
\[
\begin{aligned}
2-p_S(i)-q_S(i)
&=(1-p_S(i))+(1-q_S(i))\\
&=\frac{p(S\setminus\{i\})}{p(S)}
+\frac{q(S\setminus\{i\})}{q(S)}\\
&=\frac{p(S\setminus\{i\})q(S)+q(S\setminus\{i\})p(S)}{p(S)q(S)}.
\end{aligned}
\]
Also $p(S)+q(S)=2W$ and $p(S\setminus\{i\})+q(S\setminus\{i\})=2U$, so
\[
p(S)q(S) \leq \frac14\left(p(S)+q(S)\right)^2= W^2
\]
and
\[
p(S\setminus\{i\})q(S)+q(S\setminus\{i\})p(S) \leq (p(S)+q(S))(p(S\setminus\{i\})+q(S\setminus\{i\})) =4UW.
\]
As a last observation, $p_S(i)+q_S(i)\leq2$.
Putting this all together, 
\[
\begin{aligned}
\F(S)
&\geq\frac{(p_S(i)-q_S(i))^2}{(p_S(i)+q_S(i))(2-p_S(i)-q_S(i))}\\
	&\geq \frac{4w_i^2Y^2 \cdot p(S)q(S)}{2\cdot (p(S)q(S))^2\bigl(p(S\setminus\{i\})q(S)+q(S\setminus\{i\})p(S)\bigr)}\\
&\geq\frac{w_i^2Y^2}{2W^3U},
\end{aligned}
\]
which proves \eqref{eq:pointwise} in this case.

We use the following technical lemma to bound the expectation of the quantity in the denominator, whose proof appears in Appendix~\ref{app:moment}. 
\begin{restatable}{lemma}{momentbound}\label{lem:moment}
Let $0<\rho\leq1/2$, let $0<a\leq2\rho$, and let $b_1,\ldots,b_s\in[0,2\rho]$ satisfy
\[
T:=a+\sum_{\ell=1}^s b_\ell\leq1.
\]
Let $\xi_1,\ldots,\xi_s$ be independent $\Ber(\rho)$ variables, and set $U=\sum_\ell b_\ell\xi_\ell$. Then
\begin{equation*}\label{eq:moment}
\E[(a+U)^3U]\leq159\rho^4T.
\end{equation*}
\end{restatable}

We apply Lemma~\ref{lem:moment} with $a=w_i$ and the vector $b$ being the vector $w_\ell$ for $\ell > i$.
It is not hard to verify the lemma conditions using the definition and choice of $\rho$, the fact that $w_\ell$ is non-increasing, and $T$ corresponds to $T_i \leq 1$. 
This gives 
\begin{equation}\label{eq:moment-used}
\E_S[W^3U\mid A_i,j]\leq 159\rho^4T_i \leq 256 \rho^4 T_i.
\end{equation}

We let $Z = W^3 U$, and decompose $Y = \frac{Y}{\sqrt{Z}} \sqrt{Z}$ for the sake of applying Cauchy-Schwarz. 
Specifically: 
\[
\begin{aligned}
(\E_S[Y\mid A_i,j])^2
&=\left(\E_S\!\left[\frac{Y}{\sqrt Z}\sqrt Z\,\middle|\,A_i,j\right]\right)^2\\
&\leq\E_S[Y^2/Z\mid A_i,j]\E_S[Z\mid A_i,j].
\end{aligned}
\]
Returning to the pointwise bound \eqref{eq:pointwise} and taking (conditional) expectation of both sides: 
\begin{equation*}\label{eq:conditional-F}
\begin{aligned}
\E_S[\F(S)\mid A_i,j]
&\geq\frac{w_i^2}{2}\E_S[Y^2/Z\mid A_i,j]\\
&\geq\frac{w_i^2(\E_S[Y\mid A_i,j])^2}{2\E_S[Z\mid A_i,j]}\\
&\geq\frac{w_i^2}{512\rho^2}\,T_i(v_i-m_i)^2.
\end{aligned}
\end{equation*}
The last inequality uses \eqref{eq:EY} and \eqref{eq:moment-used}.
Finally, use $\Prob_S(A_i\mid j)\geq\rho/4$ from \cref{prop:i-heaviest} and $w_i\geq\rho$:
\[
\E_S[\F(S)\ind_{A_i}\mid j]
=\Prob_S(A_i\mid j)\E_S[\F(S)\mid A_i,j]
\geq\frac{w_i^2}{2048\rho}T_i(v_i-m_i)^2
\geq\frac{w_i}{2048}T_i(v_i-m_i)^2.
\]
\end{proof}

\subsection{Converting Expectation to Probability with Work Investment}\label{sec:investment}

Corollary~\ref{cor:expected-discrep} establishes a lower bound on the expected discrepancy of a random $S$. 
This is sufficient to conclude a (sub-optimal) algorithm for testing equivalence.
We let $\eta=\frac{\varepsilon^2}{2^{16}\left(1 + \log_2 \frac{16n}{\varepsilon}\right)\left(1+\ln\frac{16}{\varepsilon}\right)} = \tilde \Theta\left(\frac{\varepsilon^2}{\log n}\right)$.
Corollary~\ref{cor:expected-discrep} states that $\E[\F(S)] \geq 2\eta$. 
A reverse Markov's inequality says that $\Prob[\F(S) \geq \eta] \geq \frac12\E[\F(S)] \geq \eta$.
Thus, sampling $O(1/\eta)$ sets $S$ gives reasonable probability that at least one of them has $\F(S) \geq \eta$. 
Running the test of Corollary~\ref{cor:test-discrep} on each of the sampled sets with threshold $\eta$ gives a correct algorithm that requires $\tilde O(1/\eta^2) = \tilde O\left(\frac{\log^2 n}{\varepsilon^4}\right)$ samples -- the square of our target sample complexity.  

We avoid this quadratic factor through more careful analysis using Levin's work-investment strategy~\cite{Goldreich14}. 
We specifically use the refinement due to Berman, Raskhodnikova, and Yaroslavtsev~\cite{BermanRY14}. 
Intuitively, $\E[\F(S)] \geq 2\eta$ could occur due to various distributions on $\F(S)$. 
It could be since $\F(S) = 1$ with probability $2\eta$. 
In this case, we could achieve an $\tilde O\left(\frac{\log n}{\varepsilon^2}\right)$ query complexity by drawing $\tilde O\left(1/\eta\right)$ random sets $S$ and testing for discrepancy on each with threshold $\Theta(1)$. 
Alternatively, perhaps $\F(S) \geq 2\eta$ with probability $1$.
In this case as well, we can achieve our target query complexity by drawing a single random set $S$ and testing for discrepancy with threshold $2\eta$. 
Morally similar properties hold for any distribution on $\F(S)$, though we do not know beforehand which case we fall into. 
Levin's work-investment strategy iterates over a range of multiplicatively decreasing thresholds and increasing number of samples per set, effectively trying all possibilities at the cost of only a logarithmic overhead.

\begin{proof}[Proof of \cref{lem:energy}]
Instantiate the last threshold in Line~\ref{ln:thresholds} of Algorithm~\ref{alg:main} to be $\eta$. 
By \cref{cor:expected-discrep}, $\E_{j,S}[\F(S)]\geq2\eta$. 
Our first step is to discount all sets with discrepancy below $\eta$, since our algorithm never attempts to test for discrepancies below $\eta$. 
They can contribute at most $\eta$ to the expectation, and thus
\[
\E_{j,S}[\F(S)\ind_{\{\F(S)\geq \eta\}}]
\geq\E_{j,S}[\F(S)]-\eta\geq \eta.
\]
	For each threshold $t$, let
\[
	\beta_t=\Prob_{j,S}[\F(S)\geq t],
\]
and instantiate
\[
T(t)=\frac{12t}{\eta}.
\]
We compute the probability that the lemma's conclusion does not hold:
\[
	\Prob[\F(S_{t,\ell})<t\text{ for every }t,\ell]
=\prod_t(1-\beta_t)^{T(t)}
\leq\prod_t \exp\!\left(-\frac{12t}{\eta} \beta_t\right)
=\exp\!\left(-\frac{12}{\eta}\sum_t t\beta_t\right).
\]
Consider some set $S$ with discrepancy $\F(S) \geq \eta$. 
The algorithm will test some $t^*$ such that $\F(S) \geq t^* \geq \frac{1}{2}\F(S)$. 
We can represent this property as a sum of indicators, which ``picks out'' the $t^*$ for which this holds:
\[
	\sum_t t\ind_{\{\F(S)\geq t \geq \frac{1}{2}\F(S)\}}
\geq\tfrac12\F(S)\ind_{\{\F(S)\geq \eta\}}.
\]
Taking the expectation of both sides with respect to the random choice of $j$ and $S$, 
\[
	\sum_t t \Prob_{j,S}\left[\F(S) \geq t \geq \frac{1}{2}\F(S)\right]
	\geq\tfrac12\E_{j,S}[\F(S)\ind_{\{\F(S)\geq \eta\}}] \geq \eta/2.
\]
Observing that 
$\Prob_{j,S}\left[\F(S) \geq t \geq \frac{1}{2}\F(S)\right] \leq \Prob_{j,S}\left[\F(S) \geq t \right] = \beta_t$, we have 
	\[
\sum_t t \beta_t \geq \eta/2.
	\]
	Thus, with probability at least $1 - \exp\left(-\frac{12}{\eta}\cdot \frac{\eta}{2}\right) = 1-e^{-6}>9/10$, some sampled set satisfies $\F(S_{t,\ell})\geq t$, as claimed. 
\end{proof}

\section{Lower Bound}\label{sec:lower-bound}

We proceed to prove the lower bound in Theorem~\ref{thm:equiv-main}, which shows that every algorithm for equivalence testing with non-adaptive conditional samples requires $\Omega\left(\frac{\log n}{\varepsilon^2}\right)$ samples. 
In Section~\ref{sec:lb-construction}, we describe the hard instance.
The two distributions share a random support of a random size. 
When $p=q$, they are both uniform on this support, and when $\dtv(p,q) \geq \varepsilon$, they feature a Paninski-style construction~\cite{Paninski08}, where a random half of the symbols in $p$ have their probabilities multiplicatively perturbed by $1+\varepsilon$ and the other half are perturbed by $1-\varepsilon$ ($q$ is perturbed in the opposite direction). 
In Section~\ref{sec:fixed-support}, we analyze the total variation distance between the transcripts (i.e., the sequence of responses) under the two cases ($p=q$ versus $\dtv(p,q) \geq \varepsilon$) for a fixed support and set of queries: the only randomness is over the pattern of the perturbation. 
A lengthy calculation involving the square of the likelihood ratio produces a bound which depends on the sizes of the intersections of the queries with the support.
In Section~\ref{sec:avg-support}, we randomize over the support to show that a single query (in expectation) contributes little to the total variation distance between transcripts. 
Finally, in Section~\ref{sec:finish-lb}, we assemble these pieces and complete the proof by relating the transcripts' total variation distance to the test's success probability. 

\subsection{The Hard Input Distribution}\label{sec:lb-construction}
Our hard instance is defined as follows. 
A random support $D \subseteq [n]$ is chosen to be the domain for both $p$ and $q$. 
In the case where $p = q$, they are both uniform over $D$. 
In the case where $\dtv(p,q) \geq \varepsilon$, we choose a random half of $D$ and multiply its weights under $p$ and $q$ by $1+\varepsilon$ and $1-\varepsilon$, respectively, and we reverse these multiplicative factors for the other half.  

In more detail, let $L=\log_2 n$.
Choose $J$ uniformly from $\{1,\ldots,L\}$, set $s=2^J$, and choose
$D\subseteq[n]$ uniformly among all subsets of size $s$.
Write $U_D=\operatorname{Unif}(D)$ for the uniform distribution over $D$. 
Choose a sign vector $\sigma$, with exactly $s/2$ $+1$'s and $s/2$ $-1$'s, uniformly at random: 
\[
  \sigma\in\{-1,+1\}^{D},
  \qquad \sum_{x\in D}\sigma_x=0.
\]
We consider the following null hypothesis $\mathcal{H}_0$ and alternative hypothesis $\mathcal{H}_1$: 
\begin{align*}
  \mathcal H_0 &: \quad p=q=U_D,\\
  \mathcal H_1 &: \quad
  p(x)=\frac{1+\varepsilon\sigma_x}{s},\qquad
  q(x)=\frac{1-\varepsilon\sigma_x}{s}
  \quad (x\in D),
\end{align*}
with both distributions zero outside $D$.

\subsection{Comparing Answers on a Fixed Support}
\label{sec:fixed-support}

We first fix the support $D$ and a list of queries, and compare the answers in the uniform and perturbed cases. We call the entire list of answers the \emph{transcript}. 
In this section, under those restrictions, we upper bound the total variation distance between the transcripts under $\mathcal H_0$ and $\mathcal H_1$. 
In the following section, we incorporate the randomness in the choice of $D$ as well.

We use the following function $h(m)$ to quantify how much a query contributes to our bound on the transcript distance, as a function of $m$, the size of the overlap between the query set and the support $D$. 
If a query intersects $D$ at $0$ or $1$ point, the response will follow the same distribution under both hypotheses. 
Otherwise, we quantify the contribution as $1/\sqrt{m}$ -- intuitively, larger intersections contribute less to the bound due to cancellations between the random signs.

\begin{equation}\label{lb:eq:weight}
  h(m)=
  \begin{cases}
    0,&m=0\text{ or }m=1,\\
	  1/\sqrt{m},&m\ge 2.
  \end{cases}
\end{equation}

\begin{lemma}\label{lb:lem:transcript}
Fix $0<\varepsilon\le1/2$, a support $D$, and a list of queries $(S_1,b_1),\ldots,(S_N,b_N)$, where
for each $i \in [N]$, $S_i\subseteq[n]$ is the query set and $b_i\in\{p,q\}$ is the queried distribution.
Let $s_i=|S_i\cap D|$ be the number of support points in the $i$th query.
Let $\mathsf T_0^D$ and $\mathsf T_1^D$ be the distributions of the entire
answer list under $\mathcal H_0$ and $\mathcal H_1$, respectively,
conditional on the support $D$ (but not the random perturbation vector $\sigma$). 

There is an absolute constant $C$ such that
\begin{equation}\label{lb:eq:transcript-bound}
  \TV(\mathsf T_0^D,\mathsf T_1^D)
  \le \min\left\{1,\ C\varepsilon^2
                   \sum_{i=1}^N h(s_i)\right\}.
\end{equation}
\end{lemma}

\begin{proof}
We use the second-moment method to bound the total variation distance, by looking at the second moment of the likelihood ratio. 
Specifically, we upper bound the total-variation distance between the transcripts in terms of their chi-squared distance:
\[
  \TV(\mathsf T_1^D,\mathsf T_0^D)
	\leq \tfrac12\sqrt{\chi^2(\mathsf T_1^D \| \mathsf T_0^D)}
  =\tfrac12\sqrt{\E_{\mathsf T_0^D}[\mathcal L^2]-1},
\]
where $\mathcal L$ is the likelihood ratio 
\[
  \mathcal L(x)
    =\frac{\mathsf T_1^D(x)}{\mathsf T_0^D(x)}.
\]

If a query has $s_i=0$ or $1$, then the answer distribution is identical under the two hypotheses and the likelihood ratio for this query will be $1$, regardless of the sign vector $\sigma$. 
For the sequel, we focus on the remaining queries which have $s_i \geq 2$. 

Focus on the likelihood ratio for a single answer. 
For a fixed sign vector $\sigma$, let $\bar\sigma_i$ be its average
	over the support points in query $i$ -- that is, the fraction of points that are spiked upwards minus the fraction of points that are spiked downwards, which gives a measure of the excess mass assigned to $p$ under $\mathcal{H}_1$. 
\[
  \bar\sigma_i=\frac1{s_i}\sum_{x\in S_i\cap D}\sigma_x,
\]
We let $r_{\sigma,i}(x)$ be the
ratio of the probability of answering $x$ under $\mathcal H_1$ to that
under $\mathcal H_0$.
For distribution $p$ and under $\mathcal{H}_0$, the probability of outputting any $x \in S_i \cap D$ is $1/s_i$. 
Under $\mathcal{H}_1$, the probability is $\frac{1 + \varepsilon\sigma_x}{s_i + \varepsilon s_i \bar \sigma_i}.$
Taking their ratio (and with a similar calculation for $q$), for any $x \in S_i \cap D$, we have:
\[
  r_{\sigma,i}(x)
  =\begin{cases}
    \displaystyle\frac{1+\varepsilon\sigma_x}{1+\varepsilon\bar\sigma_i},&b_i=p,\\[6pt]
    \displaystyle\frac{1-\varepsilon\sigma_x}{1-\varepsilon\bar\sigma_i},&b_i=q.
  \end{cases}
\]
Note that these likelihood ratios are conditioned on a particular $\sigma$, whereas the transcript distribution $\mathsf T_1^D$ chooses $\sigma$ uniformly at random, which averages the likelihood ratio over all $\sigma$.  
For any $x \in S_i \cap D$, a symmetry argument implies that $\E_\sigma[p(x \mid S_i)] = \E_\sigma[q(x \mid S_i)] = 1/s_i$.
	Since the same is true for a sample under $\mathcal{H}_0$ (since the distribution is uniform), we have that 
\begin{equation}
	\E_\sigma[r_{\sigma,i}(x)]=1. \label{lb:eq:one-coordinate}
\end{equation}

The challenge in analysis is that we are doing \emph{several} samples with respect to a single sign vector $\sigma$, not just one sample, and we must account for correlations between the likelihood ratios of individual samples. 
The likelihood ratio for the whole transcript is as follows: 
\[
  \mathcal L(x)
  =\E_\sigma\left[\prod_i r_{\sigma,i}(x_i)\right].
\]
Note that, conditioned on a fixed sign vector $\sigma$, the answers are independent, justifying the product of the $r_{\sigma,i}$'s inside the expectation. 

We look at the square of the likelihood ratio by taking a product with the same expression using independent randomness; we let $\sigma'$ denote an independently sampled sign vector.  
\[
  \mathcal L(x)^2
  =\E_{\sigma}\left[\prod_i r_{\sigma,i}(x_i)\right]\E_{\sigma'}\left[\prod_i r_{\sigma',i}(x_i)\right]
  =\E_{\sigma,\sigma'}\left[\prod_i r_{\sigma,i}(x_i)r_{\sigma',i}(x_i)\right].
\]

Recall our goal is to bound 
	\begin{equation*}
	\begin{aligned}
		\E_{x \sim \mathsf T_0^D}[\mathcal L^2] &= \E_{x \sim \mathsf T_0^D}\left[\E_{\sigma,\sigma'}\left[\prod_i r_{\sigma,i}(x_i)r_{\sigma',i}(x_i)\right]\right] \\ 
		&= \E_{\sigma,\sigma'}\left[\E_{x \sim \mathsf T_0^D}\left[\prod_i r_{\sigma,i}(x_i)r_{\sigma',i}(x_i)\right]\right] \\
		&= \E_{\sigma,\sigma'}\left[\prod_i \E_{x_i \sim \operatorname{Unif}(S_i \cap D)} [r_{\sigma,i}(x_i)r_{\sigma',i}(x_i)]\right].
	\end{aligned}
	\end{equation*}
	We can factor the expectation in the last equality since the $x_i$'s are independent under the null. 

	We rewrite each individual factor in this product as 
	\[
		1 + K_i(\sigma, \sigma') := \E_{x_i \sim \operatorname{Unif}(S_i \cap D)} [r_{\sigma,i}(x_i)r_{\sigma',i}(x_i)].
	\]
	We proceed to derive an expression for 
	\[
	K_i(\sigma,\sigma') = \E_{x_i \sim \operatorname{Unif}(S_i \cap D)} [r_{\sigma,i}(x_i)r_{\sigma',i}(x_i)] -1. 
	\]
We illustrate by focusing on a query to $p$; a query to $q$ will be analogous. 
Fix the sign vectors $\sigma, \sigma'$. 
We substitute the expression for $r_{\sigma,i}$ and $r_{\sigma',i}$ and expand the expectation: 
\[
  K_i(\sigma,\sigma')
  =\frac1{s_i}\sum_{x_i\in S_i\cap D}
    \frac{(1+\varepsilon\sigma_{x_i})(1+\varepsilon\sigma'_{x_i})}
         {(1+\varepsilon\bar\sigma_i)(1+\varepsilon\bar{\sigma}'_i)}-1 
  =\frac1{s_i}\sum_{x_i\in S_i\cap D}
  \frac{(1+\varepsilon\sigma_{x_i})(1+\varepsilon\sigma'_{x_i})
          -(1+\varepsilon\bar\sigma_i)(1+\varepsilon\bar{\sigma}'_i)}
         {(1+\varepsilon\bar\sigma_i)(1+\varepsilon\bar{\sigma}'_i)}.
\]
We expand the numerator as follows:
\begin{align*}
  &(1+\varepsilon\sigma_{x_i})(1+\varepsilon\sigma'_{x_i})
    -(1+\varepsilon\bar\sigma_i)(1+\varepsilon\bar{\sigma}'_i)\\
  &\qquad=\varepsilon(\sigma_{x_i}+\sigma'_{x_i}
                           -\bar\sigma_i-\bar{\sigma}'_i)
       +\varepsilon^2(\sigma_{x_i}\sigma'_{x_i}
                           -\bar\sigma_i\bar{\sigma}'_i).
\end{align*}
The terms which are linear in $\varepsilon$ vanish when we average over $x_i \in S_i \cap D$, as we substitute  $\bar\sigma_i=\frac1{s_i}\sum_{x_i \in S_i\cap D}\sigma_{x_i}$
\[
  \frac1{s_i}\sum_{x_i\in S_i\cap D}
    (\sigma_{x_i}+\sigma'_{x_i}-\bar\sigma_i-\bar{\sigma}'_i)
  =\bar\sigma_i+\bar{\sigma}'_i-\bar\sigma_i-\bar{\sigma}'_i=0.
\]
Only the terms which are quadratic in $\varepsilon$ remain, yielding
\[
  K_i(\sigma,\sigma')
  =\varepsilon^2\,
    \frac{\displaystyle\frac1{s_i}\sum_{x_i\in S_i\cap D}
                         \sigma_{x_i}\sigma'_{x_i}
          -\bar\sigma_i\bar{\sigma}'_i}
         {(1+\varepsilon\bar\sigma_i)(1+\varepsilon\bar{\sigma}'_i)}.
\]
Doing the same calculations for a query to $q$, we obtain an expression where the only difference is the denominator becomes $(1-\varepsilon\bar\sigma_i)(1-\varepsilon\bar{\sigma}'_i)$.

At this point, we have obtained:
\begin{equation}\label{lb:eq:chi-identity}
  \E_{\mathsf T_0^D}[\mathcal L^2]
  =\E_{\sigma,\sigma'}\left[\prod_i(1+K_i(\sigma,\sigma'))\right].
\end{equation}
We further upper bound this using the inequality $1 + x \leq e^x$: 

\begin{equation*}
  \E_{\mathsf T_0^D}[\mathcal L^2]
	\leq\E_{\sigma,\sigma'}\left[\prod_i\exp\left(K_i(\sigma, \sigma')\right)\right]
	=\E_{\sigma,\sigma'}\left[\exp\left(\sum_i K_i(\sigma, \sigma')\right)\right] 
	= \E_{\sigma,\sigma'}\left[e^Z\right],
\end{equation*}
where we have defined
\[
  Z=\sum_i K_i(\sigma,\sigma').
\]
We use the power series expansion of this exponential:
\[
  \E_{\sigma,\sigma'}\left[e^Z\right]
  = \sum_{k=0}^\infty \E_{\sigma,\sigma'}\left[\frac{Z^k}{k!}\right].
\]
The $k = 0$ term evaluates to $1$. 
The $k= 1$ term $\E_{\sigma, \sigma'}[Z]=\sum_i \E_{\sigma, \sigma'} [K_i(\sigma,\sigma')]$ vanishes:
\[
	\E_{\sigma,\sigma'}[K_i(\sigma,\sigma')] = \E_{x_i \sim \operatorname{Unif}(S_i \cap D)} [\E_\sigma [r_{\sigma,i}(x_i)]\E_{\sigma'} [r_{\sigma',i}(x_i)]]  -1 = 1^2 -1 = 0,
\]
where the first equality uses the definition of $K_i$ and independence of $\sigma$ and $\sigma'$, and the second equality uses \eqref{lb:eq:one-coordinate}. 
We then sum over $i$ to get $\E_{\sigma, \sigma'}[Z] = 0$.

Consequently,
\[
  \E_{\sigma,\sigma'}\left[e^Z\right]
  =1+\sum_{k=2}^\infty\frac{\E_{\sigma,\sigma'}[Z^k]}{k!}
  \le1+\sum_{k=2}^\infty\frac{\E_{\sigma,\sigma'}[|Z|^k]}{k!}
  =1+\sum_{k=2}^\infty\frac{\|Z\|_k^k}{k!}.
\]
It remains to bound these absolute moments. We do this by bounding
the moments of each $K_i$ and then applying the triangle inequality.
The proof of the following lemma (a variant of Khintchine inequality) appears in Appendix~\ref{app:balanced-signs}. 

\begin{restatable}{lemma}{balancedsignbound}\label{lb:lem:signs}
Let $D$ have even cardinality, and let $\sigma$ be uniformly sampled from the set of all balanced sign vectors on $D$. 
There is an absolute constant $C_0$ such that, for
every vector $a\in\mathbb{R}^D$ and every $k\ge 2$,
\[
  \left\|\sum_{x\in D}a_x\sigma_x\right\|_k
  \le C_0\sqrt k\left(\sum_{x\in D}a_x^2\right)^{1/2},
  \qquad
  \left\|X\right\|_k=(\E[|X|^k])^{1/k}.
\]
\end{restatable}
Specifically, we apply this bound to the two terms in the numerator of $K_i(\sigma, \sigma')$:
\[
  \frac1{s_i}\sum_{x\in S_i \cap D}\sigma_x\sigma'_x
  \qquad \text{and} \qquad
  \bar\sigma_i\bar{\sigma}'_i.
  \]

For the first term, fix $\sigma'$ and choose the coefficients in
Lemma~\ref{lb:lem:signs} to be
\[
  a_x=\begin{cases}
    \sigma'_x/s_i,&x\in S_i\cap D,\\
    0,&x\in D\setminus S_i.
  \end{cases}
\]
Then
\[
  \sum_{x\in D}a_x\sigma_x
  =\frac1{s_i}\sum_{x\in S_i\cap D}\sigma_x\sigma'_x,
  \qquad
  \sum_{x\in D}a_x^2
  =\frac{s_i}{s_i^2}=\frac1{s_i}.
\]
Thus, for every integer $k\ge2$, the lemma gives
\[
  \E_\sigma\left[
    \left|\frac1{s_i}\sum_{x\in S_i\cap D}\sigma_x\sigma'_x\right|^k
    \,\middle|\,\sigma'\right]
  \le\left(C_0\sqrt{\frac{k}{s_i}}\right)^k.
\]
This bound holds for every fixed $\sigma'$, so averaging over $\sigma'$
gives the same bound on the joint expectation over both sign vectors.

For the second term, choose $a_x=1/s_i$ on $S_i\cap D$ and zero
elsewhere in $D$. Now $\sum_xa_x\sigma_x=\bar\sigma_i$, and the
squared coefficients again sum to $1/s_i$. The lemma gives
\[
  \E_\sigma[|\bar\sigma_i|^k]
  \le\left(C_0\sqrt{\frac{k}{s_i}}\right)^k.
\]
The same bound applies to $\bar{\sigma}'_i$. Independence of the two
sign vectors then gives
\begin{align*}
  \E_{\sigma,\sigma'}[|\bar\sigma_i\bar{\sigma}'_i|^k]
  &=\E_\sigma[|\bar\sigma_i|^k]\,
    \E_{\sigma'}[|\bar{\sigma}'_i|^k]\\
  &\le\left(C_0^2\frac{k}{s_i}\right)^k.
\end{align*}
Taking $k$th roots, we obtain the following bounds, with both norms
over the two independent sign vectors:
\[
  \left\|\frac1{s_i}\sum_{x\in S_i\cap D}\sigma_x\sigma'_x\right\|_k
    \le C_0\sqrt{\frac k{s_i}},
  \qquad
  \left\|\bar\sigma_i\bar{\sigma}'_i\right\|_k
    =\left\|\bar\sigma_i\right\|_k\left\|\bar{\sigma}'_i\right\|_k
    \le C_0^2\frac k{s_i}.
\]

We turn our attention to the denominator of $K_i(\sigma, \sigma')$, which is $(1 \pm \varepsilon \bar \sigma_i)(1 \pm \varepsilon \bar \sigma_i')$, where the sign depends on whether the query is to $p$ or $q$. 
Regardless, based on the range of $\bar \sigma_i$, we can lower bound this expression by $(1 - \varepsilon)^2 \geq 1/4$. 
Using this denominator bound and then the triangle inequality gives
\begin{equation}\label{lb:eq:kernel-moment}
\begin{aligned}
  \|K_i\|_k
  &\le4\varepsilon^2
    \left\|\frac1{s_i}\sum_{x_i\in S_i\cap D}\sigma_{x_i}\sigma'_{x_i}
           -\bar\sigma_i\bar{\sigma}'_i\right\|_k\\
  &\le4\varepsilon^2\left(
    \left\|\frac1{s_i}\sum_{x_i\in S_i\cap D}\sigma_{x_i}\sigma'_{x_i}\right\|_k
    +\|\bar\sigma_i\bar{\sigma}'_i\|_k\right)\\
  &\le4\varepsilon^2\left(C_0\sqrt{\frac{k}{s_i}}
                          +C_0^2\frac{k}{s_i}\right)\\
  &\le C_1\varepsilon^2\frac{k}{\sqrt{s_i}}
   =C_1\varepsilon^2k\,h(s_i).
\end{aligned}
\end{equation}
The last inequality uses $\sqrt{k}\le k$, $1/s_i\le1/\sqrt{s_i}$, and introduces a new absolute constant $C_1 = 4(C_0 + C_0^2)$.
Observe here that the dependence on the size of the intersection $s_i$ is $1/\sqrt{s_i}$, which is precisely how we defined the function $h(s_i)$. 

We can now bound the moments of $Z$.
\begin{equation}\label{lb:eq:Z-moment}
  \|Z\|_k
  =\left\|\sum_iK_i\right\|_k
  \le\sum_i\|K_i\|_k
  \le\sum_i C_1\varepsilon^2k\,h(s_i)
  =C_1\varepsilon^2k\sum_i h(s_i)=\gamma k,
\end{equation}
where we let $\gamma = C_1 \varepsilon^2 \sum_i h(s_i)$. 
The two inequalities are triangle inequality and \eqref{lb:eq:kernel-moment}.

We now return to the series we were trying to bound.
First, we consider the case $e \gamma \leq 1/2$. 
Equation~\eqref{lb:eq:Z-moment} gives
$\E_{\sigma,\sigma'}[|Z|^k]\le(\gamma k)^k$, and $k!\ge(k/e)^k$
turns the bound into a geometric series:
\begin{align*}
  \E_{\mathsf T_0^D}[\mathcal L^2]
  &\le \E_{\sigma,\sigma'}[e^Z]\\
  &\le 1+\sum_{k=2}^\infty\frac{\E_{\sigma,\sigma'}[|Z|^k]}{k!}\\
  &\le 1+\sum_{k=2}^\infty(e\gamma)^k
   \le 1+2e^2\gamma^2.
\end{align*}
The TV inequality from the start of the proof now gives
\[
  \TV(\mathsf T_1^D,\mathsf T_0^D)
  \le\tfrac12\sqrt{\E_{\mathsf T_0^D}[\mathcal L^2]-1}
  \le e\gamma/\sqrt{2}= \frac{eC_1}{\sqrt{2}}\varepsilon^2 \sum_i h(s_i),
\]
as desired by the lemma statement.
When $e\gamma>1/2$, we can use the trivial bound $\TV\le1\le2e\gamma$ to conclude the proof. 
\end{proof}

\subsection{Averaging over the Support Size}
\label{sec:avg-support}
The previous section assumed a fixed support $D$.
In this section, we incorporate the randomness in the choice of $D$. 
Informally speaking, a query set $S$ will only be informative if the size of its intersection with $D$ is non-zero (more precisely, at least $2$) but also not too large (as the contribution to the total variation bound decays as $1/\sqrt{|S \cap D|}$). 
Thus, a fixed query set $S$ would only contribute significantly under a few choices of the size of $D$ (on average, accounting for the random selection of $D$ itself). 
Since our hard instance construction randomizes over $\log n$ scales of the size of $D$, unless we take at least that many queries, the distance between transcripts under $\mathcal H_0$ and $\mathcal H_1$ remains small.



\begin{lemma}
\label{lb:lem:scales}
	Let $L = \log_2 n$, choose $J$ uniformly from $\{1, 2, \dots, \log_2 n\}$, and set $s = 2^J$. Choose $D \subseteq [n]$ to be a uniformly random subset of $[n]$ of size $s$. 
	For every fixed set
$S\subseteq[n]$,
\begin{equation}\label{lb:eq:scale-average}
  \E_D[h(\left|S\cap D\right|)]\le\frac{20}{L}.
\end{equation}
\end{lemma}

\begin{proof}
Let $b=|S|\ge1$.
For the time being, fix the support size $s = 2^j$ and 
	define
\[
  M=|S\cap D|,\qquad \lambda=\E[M]=bs/n.
\]
The randomness in the following bounds is exclusively over a uniformly random support of this fixed size, and we incorporate randomness in $s$ itself at the very end of the proof. 

We bound $\E[h(M)]$ separately for small and large values of $\lambda$.
We start with $\lambda\le1$: intuitively, this captures the case where the query is unlikely to contain enough support points to contribute to the expectation. 
Since $h(M)=0$ when $M=0$ or $1$, and $h(M)\le1\le M/2$
when $M\ge2$, we have $h(M)\le M/2$. Taking expectations gives
\begin{equation}\label{lb:eq:small-scale}
  \E[h(M)]\le\frac{\E[M]}2=\frac\lambda2.
\end{equation}

If $\lambda>1$ (i.e., the expected intersection is substantial), split according to whether $M\ge\lambda/2$.
If so, then $h(M)$ will be small: $h(M) \leq \sqrt{2/\lambda}$. 
In the other case, $h(M)$ may be large (we again will use the crude bound $h(M) \leq 1$), but we bound the probability that this case happens using Chebyshev's inequality. 
The indicators $\mathbf1_{\{x\in D\}}$, for $x\in S$, are negatively
correlated random variables: conditioning on one point belonging to the
	fixed-size support makes another less likely to belong. Hence, noting that $M = \sum_{x \in S}\mathbf1_{\{x\in D\}}$
\[
  \operatorname{Var}(M)
  \le\sum_{x\in S}\operatorname{Var}(\mathbf1_{\{x\in D\}})
  =\sum_{x\in S}\frac{s}{n}\left(1-\frac{s}{n}\right)
  \le\sum_{x\in S}\frac{s}{n}
  =\frac{bs}{n}=\lambda.
\]
Chebyshev's inequality then gives
$\Prob(M<\lambda/2)\le4/\lambda$. Consequently,
\begin{equation}\label{lb:eq:large-scale}
  \E[h(M)]
  \le\sqrt{\frac2\lambda}+\Prob(M<\lambda/2)
  \le\sqrt{\frac2\lambda}+\frac4\lambda
  \le\frac{\sqrt2+4}{\sqrt\lambda}.
\end{equation}

At this point, we sum over the choice of the support size $s = 2^j$. 
The expected intersection size
$\lambda_j=b2^j/n$ doubles at each successive scale.
We split the analysis as we did before, when $\lambda_j \leq 1$ and when $\lambda_j > 1$. 
	In the former case, a geometric series bounds the sum of conditional expectations $\sum_{j : \lambda_j \leq 1} \E\bigl[h(|S\cap D|)\mid s = 2^j]$ by $\frac{1}{2} \sum_{k \geq 0} \left(\frac12\right)^{k} = 1$. 
	In the latter case, we can similarly bound the sum of conditional expectations $\sum_{j : \lambda_j > 1} \E\bigl[h(|S\cap D|)\mid s = 2^j]$ by the geometric series $(\sqrt{2} + 4) \sum_{k \geq 0} \left(\frac{1}{\sqrt{2}}\right)^{k} = \frac{\sqrt{2}+4}{1 - 2^{-1/2}} \leq 19$.
Together,
\[
  \sum_{j=1}^L
    \E\bigl[h(|S\cap D|)\mid s=2^j\bigr]
  \le1+\frac{\sqrt2+4}{1-2^{-1/2}}\leq20.
\]
The support scale is uniform among these $L$ choices, so dividing by
$L$ gives the desired bound on expectation.
\end{proof}

\subsection{Completing the lower bound}
\label{sec:finish-lb}
At this point, we complete the argument by assembling the two pieces from the previous sections. 

\begin{proof}[Proof of the lower bound of Theorem~\ref{thm:equiv-main}]
	Suppose an algorithm makes $N$ queries, and fix its random seed $\omega$. 
	We represent its query list as 
\[
  (S_1(\omega),b_1(\omega)),\ldots,(S_N(\omega),b_N(\omega)),
\]
where $S_i$ denotes the $i$-th query set and $b_i \in \{p,q\}$ denotes which oracle we query. 
Note that since the algorithm is non-adaptive, the query list is independent of the input. 

We augment the transcript by including the identity of the support $D$ and the random seed $\omega$. 
This can only give the testing algorithm more information and thus will not increase its error probability.
Conditioning on $(\omega, D)$, the answer lists under $\mathcal H_0$ and
$\mathcal H_1$ are exactly those compared in
Lemma~\ref{lb:lem:transcript}, where the fixed query list is $(S_i(\omega),b_i(\omega))_{i=1}^{N}$.
Let $\mathsf T_0,\mathsf T_1$ denote the distribution of the augmented transcripts, 
and let $\mathsf T_0^{D,\omega},\mathsf T_1^{D,\omega}$ denote the distribution of the augmented transcripts conditioned on a particular support $D$ and query list determined by a random seed $\omega$. 
Because the augmented transcripts have the
same distribution of $(\omega,D)$ under both $\mathcal H_0$ and $\mathcal H_1$,
\[
  \TV(\mathsf T_0,\mathsf T_1)
  =\E_{\omega,D}\left[
      \TV(\mathsf T_0^{D,\omega},\mathsf T_1^{D,\omega})\right].
\]
Applying Lemma~\ref{lb:lem:transcript} for each fixed $(D,\omega)$,
then using linearity of expectation and Lemma~\ref{lb:lem:scales}, gives
\begin{align}
  \TV(\mathsf T_0,\mathsf T_1)
  &\le C\varepsilon^2\,
       \E_{\omega,D}\left[\sum_{i=1}^{N}
          h(\left|S_i(\omega)\cap D\right|)\right]\nonumber\\
  &= C\varepsilon^2\,
       \E_\omega\left[\sum_{i=1}^{N}
          \E_D[h(\left|S_i(\omega)\cap D\right|)]\right]\nonumber\\
  &\le C\varepsilon^2\,
       \E_\omega\left[\sum_{i=1}^{N}\frac{20}{L}\right]\nonumber\\
  &= \frac{20C\varepsilon^2N}{L}.
     \label{lb:eq:final-tv}
\end{align}

Correctness requires us to output that $p=q$ with probability at least $2/3$ under $\mathcal H_0$ and with probability at most $1/3$ under $\mathcal H_1$. 
These correspond to cases when we are given transcripts drawn from distributions $\mathsf T_0$ and $\mathsf T_1$, respectively.
We can lower bound the total variation distance between these cases by this gap in probabilities, and thus
\[
  \frac13\le\TV(\mathsf T_0,\mathsf T_1)
       \le\frac{20C\varepsilon^2N}{L}.
\]
It follows that
\[
  N\ge\frac{L}{60C\varepsilon^2}
    =\Omega\!\left(\frac{\log n}{\varepsilon^2}\right),
\]
as desired.
\end{proof}

\bibliographystyle{alpha}
\bibliography{biblio}

\appendix

\section{Testing Primitives}

\subsection{\textsc{Test-equal} Pseudocode}
\label{sec:test-equal}

The pseudocode for \textsc{Test-equal} appears in Algorithm~\ref{alg:test-equal}. 
Note that this procedure employs the common technique of \emph{Poissonization}: rather than drawing exactly $m$ samples from a distribution, we draw a random $\Poi(m)$ samples.
This is convenient for the sake of analysis, as the number of occurrences of each symbol will be \emph{independently} distributed as $\Poi(m\cdot \alpha)$, where $\alpha$ is the probability of that symbol. 
With high probability, the total number of samples used by \textsc{Test-equal} is $O\left(\frac{\log (1/\delta)}{t}\right)$ -- consequently, we can simply increase our (fixed) sample budget by a constant factor, and if this budget is exceeded due to the Poissonization (a rare event), we can answer arbitrarily and fold any incorrect answer into the failure probability $\delta$. 

\begin{algorithm}[H]
\caption{\textsc{Test-equal}~\cite{FalahatgarJOPS15}}
\label{alg:test-equal}
\begin{algorithmic}[1]
\Require error bound $t$, error probability $\delta$, distributions $\Ber(u)$ and $\Ber(v)$.
\Statex Repeat $18 \log \frac{1}{\delta}$ times and output the majority:
\State
	Let $m' = \Poi(m)$ and $m'' = \Poi(m)$ be two independent Poisson variables with mean $m = O(1/t)$.
\State
Draw samples $x_1,x_2\ldots x_{m'}$
from the first distribution and
$y_1,y_2\ldots y_{m''}$ from the second one.
\State
Let $m_1 = \sum_{i=1}^{m'} x_i$ and $m_2 = \sum_{i=1}^{m''} y_i$.
\State If
$\frac{(m_1-m_2)^2-m_1-m_2}{m_1+m_2-1}
+\frac{((m'-m_1)-(m''-m_2))^2-(m'-m_1)-(m''-m_2)}{(m'-m_1)+(m''-m_2)-1}
\leq\frac{mt}{2}$,
then output \texttt{same}, else \texttt{diff}.
\end{algorithmic}
\end{algorithm}

\subsection{Proof of Corollary~\ref{cor:test-discrep}}
\label{sec:test-discrep}
The algorithm first draws $d = O\left(\frac{\log (1/\delta)}{t}\right)$ samples from both $p$ and $q$.
It then runs \textsc{Test-equal} on each of the observed symbols, of which there are at most $2d$, with parameters $t$ and $\delta' = \frac{\delta}{2d}$. 
To run these tests, it draws a fresh batch of samples of size $O\left(\frac{\log (1/\delta')}{t}\right) = O\left(\frac{\log(1/t) + \log (1/\delta)}{t}\right)$, which is reused across all tests. 
If all of the tests return \texttt{same}, then return \texttt{same}, if any of the tests return \texttt{diff}, then return \texttt{diff}. 

If $p = q$, then a union bound over all $\leq 2d$ tests implies that, with probability $\geq 1 - \delta$, they will all return \texttt{same} and thus the test correctly outputs \texttt{same}. 
On the other hand, suppose there exists some $i^*$ such that
$\frac{(p(i^*)-q(i^*))^2}{(p(i^*)+q(i^*))(2-p(i^*)-q(i^*))} \geq t$.
Algebraic manipulation implies that $p(i^*)+q(i^*) \geq t$, and thus the probability that $i^*$ is not observed in either of the two batches is at most $(1 - p(i^*))^d (1-q(i^*))^d \leq e^{-t d} \leq \delta/2$.
The test on symbol $i^*$ will correctly output \texttt{diff} with probability at least $1 - \delta/2d$ and thus the test outputs \texttt{diff} with probability at least $1 - \delta/2 - \delta/2d$.

\section{A Moment Bound for Random Restrictions}\label{app:moment}

We restate the lemma for convenience:

\momentbound*
\begin{proof}
	We rescale all quantities by units of $\rho$: let $a' = a/\rho$, $b' = b/\rho$, and $U' = U/\rho = \sum_\ell b'_\ell \xi_\ell$. 
Define
\[
\tau=\E[U']=\sum_\ell b'_\ell \rho = \sum_\ell b_\ell =T-a\leq T\leq1,
	\qquad s_r=\sum_\ell\rho b'^r_\ell.
\]
	Since $b'_\ell\leq2$, $s_r = \sum_\ell b'^{r-1}_\ell b'_\ell \rho \leq2^{r-1}\tau$. 
	We proceed to bound the first four moments of $U'$. 
	To do so, we can expand the powers of $U'$ and, using multinomial expansion, group terms by which indices coincide. 
	To illustrate, if we were to expand $U'^2 = \left(\sum_\ell b'_\ell \xi_\ell \right)^2$, we would have the terms where the indices do match, $\sum_\ell  b'^2_\ell \xi_\ell^2 $, and don't match, $\sum_\ell \sum_{k\neq \ell} b'_\ell b'_k \xi_\ell \xi_k$. 
	Taking expectation of the first sum, it comes to $\sum_\ell \rho b'^2_\ell = s_2$, which is upper bounded by $2\tau$. 
	Taking the expectation of the second sum and using independence of the $\xi$'s, we get $\sum_\ell \sum_{k\neq \ell} \rho^2 b'_\ell b'_k$.
	We can upper bound this by the unrestricted sum $(\sum_\ell \rho b'_\ell)^2 = s_1^2 = \tau^2 \leq \tau$.
	Other cases for the higher moments follow similarly.

\begin{align*}
\E[U']&=s_1=\tau,\\
\E[U'^2]&\leq s_2+s_1^2
\leq2\tau+\tau^2\leq3\tau,\\
\E[U'^3]&\leq s_3+3s_2s_1+s_1^3
\leq4\tau+6\tau^2+\tau^3\leq11\tau,\\
\E[U'^4]&\leq s_4+4s_3s_1+3s_2^2+6s_2s_1^2+s_1^4\\
&\leq8\tau+28\tau^2+12\tau^3+\tau^4\leq49\tau.
\end{align*}
Therefore, using $a'\leq2$,
\begin{align*}
\rho^{-4}\E[(a+U)^3U]
&=\E[(a'+U')^3U']\\
&\leq\E[(2+U')^3U']\\
&=8\E[U']+12\E[U'^2]+6\E[U'^3]+\E[U'^4]\\
&\leq(8+36+66+49)\tau=159\tau\leq159T.
\end{align*}
\end{proof}
\section{A Moment Bound for Balanced Signs}\label{app:balanced-signs}

In this section, we prove the following technical lemma:

\balancedsignbound*

Note that the challenge here comes from the balance requirement, as that introduces dependence between the signs. 
Otherwise, if the $\sigma_x$ were independent Rademacher random variables, this would be recognizable as the Khintchine inequality. 
\begin{lemma}[Khintchine inequality~\cite{Vershynin18}]\label{lem:khintchine}
Let $D$ be a finite set, and let the signs $(\xi_x)_{x\in D}$ be
independently and uniformly sampled from $\{-1,+1\}$.
There is an absolute constant $C$ such that, for
 every vector $a\in\mathbb{R}^D$ and every $k\ge2$,
\[
  \left\|\sum_{x\in D}a_x\xi_x\right\|_k
  \le C\sqrt k\left(\sum_{x\in D}a_x^2\right)^{1/2},
  \qquad
  \left\|X\right\|_k=(\E[|X|^k])^{1/k}.
\]
\end{lemma}
We essentially reduce to this statement by introducing a random perfect matching over $D$, where opposite signs are assigned to endpoints of the matching. 

\begin{proof}[Proof of Lemma~\ref{lb:lem:signs}]
Choose a uniformly random perfect matching between the elements of the set $D$.
For each pair $\{u,v\}$ in this matching, independently set
$(\sigma_u,\sigma_v)=(-1,+1)$ or $(+1,-1)$ with equal probability.
The resulting sign vector $\sigma$ satisfies the conditions of the lemma. 

Conditional on a fixed matching, each pair $\{u,v\}$ contributes
$\pm(a_u-a_v)$ to the sum $\sum_xa_x\sigma_x$, with independent uniform
	choices of sign across pairs. Applying the Khintchine inequality (Lemma~\ref{lem:khintchine})
to this Rademacher sum gives
\[
  \left(\E\left[\left|\sum_xa_x\sigma_x\right|^k
      \,\middle|\,\text{matching}\right]\right)^{1/k}
  \le C\sqrt{k}
    \left(\sum_{\{u,v\}\text{ matched}}(a_u-a_v)^2\right)^{1/2}.
\]
Since $(a_u-a_v)^2\le2(a_u^2+a_v^2)$ and each point belongs to one pair,
\[
  \sum_{\{u,v\}\text{ matched}}(a_u-a_v)^2
  \le2\sum_{x\in D}a_x^2.
\]
We raise the conditional moment bound to the $k$th power, and note that the bound holds pointwise for every matching, giving
\[
  \E\left[\left|\sum_xa_x\sigma_x\right|^k\right]
  \le\left(C\sqrt{2k}\left(\sum_{x\in D}a_x^2\right)^{1/2}\right)^k.
\]
Taking $k$th roots proves the lemma with $C_0=\sqrt{2}C$.
\end{proof}

\end{document}